\documentclass[11pt]{article}
\usepackage{amsmath,amssymb,bbm,amsthm}
\usepackage{fullpage}
\usepackage{thm-restate,color,xcolor,xspace}
\usepackage{hyperref,cleveref}
\usepackage{graphicx}
\usepackage{algorithm, algpseudocode}
\usepackage[shortlabels]{enumitem}

\usepackage{float}

\newtheorem{theorem}{Theorem}
\newtheorem{lemma}[theorem]{Lemma}
\newtheorem{definition}[theorem]{Definition}
\newtheorem{corollary}[theorem]{Corollary}

\newcommand{\R}{\mathbb{R}}

\newcommand{\ceil}[1]{\lceil#1\rceil}

\newcommand{\nnz}{\textup{nnz}}
\newcommand{\tr}{\textup{tr}}

\newcommand{\1}{\mathbbm{1}}

\newcommand{\todo}{\text{\color{red} (- TODO -)}}

\DeclareFontEncoding{LS1}{}{}
\DeclareFontSubstitution{LS1}{stix}{m}{n}
\DeclareSymbolFont{stixletters}{LS1}{stix}{m}{it}
\DeclareMathAccent{\cev}{\mathord}{stixletters}{"91}
\DeclareMathAccent{\vec}{\mathord}{stixletters}{"92}
\DeclareMathAccent{\vecev}{\mathord}{stixletters}{"95}

\renewcommand{\overleftrightarrow}{\vecev}

\begin{document}

\title{A Simple and Fast Algorithm for Fair Cuts}
\author{
    Jason Li\thanks{Carnegie Mellon University, \texttt{jmli@cs.cmu.edu}} \and
    Owen Li\thanks{Carnegie Mellon University, \texttt{tianwei2@andrew.cmu.edu}}
}
\date{\today}
\maketitle

\abstract{
We present a simple and faster algorithm for computing \emph{fair cuts} on undirected graphs, a concept introduced in recent work of Li~et~al.~(SODA 2023). Informally, for any parameter $\epsilon>0$, a $(1+\epsilon)$-fair $(s,t)$-cut is an $(s,t)$-cut such that there exists an $(s,t)$-flow that uses $1/(1+\epsilon)$ fraction of the capacity of \emph{every} edge in the cut. Our algorithm computes a $(1+\epsilon)$-fair cut in $\tilde O(m/\epsilon)$ time, improving on the $\tilde O(m/\epsilon^3)$ time algorithm of Li~et~al.\ and matching the $\tilde O(m/\epsilon)$ time algorithm of Sherman~(STOC 2017) for standard $(1+\epsilon)$-approximate min-cut.

Our main idea is to run Sherman's approximate max-flow/min-cut algorithm iteratively on a (directed) \emph{residual} graph. While Sherman's algorithm is originally stated for undirected graphs, we show that it provides guarantees for directed graphs that are good enough for our purposes.
}

\section{Introduction}

The $(s,t)$-min-cut and $(s,t)$-max-flow problems are among the most basic and well-studied problems in combinatorial optimization. A long line of research on fast algorithms~\cite{ford1956maximal,Dinitz70,GoldbergR98,LiuS20} culminated in the recent breakthrough $m^{1+o(1)}$ time algorithm of Chen~et~al.~\cite{ChenKLPGS22}. A separate line of research has focused on applying max-flow to solve other cut-based problems in combinatorial optimization, most notably Steiner min-cut~\cite{LiP20}, Gomory-Hu tree~\cite{abboud2023all}, and expander decomposition~\cite{SaranurakW19}. Using the algorithm of Chen~et~al.\ as a black box, all of these problems are now solvable in $m^{1+o(1)}$ time, which is optimal up to the factor $m^{o(1)}$.

On the other hand, the algorithm of Chen~et~al.\ (and subsequent improvements~\cite{van2023deterministic}) have a few downsides. First, the algorithms do not quite achieve ``near''-linear time, which colloquially means $\tilde O(m)$ time where $\tilde O(\cdot)$ suppresses polylogarithmic factors. In fact, a near-linear time algorithm appears out of reach with the current techniques, which exploit recursion at the cost of $m^{o(1)}$ factors. Also, the algorithms are inherently sequential, leaving open the question of parallel max-flow in $m^{1+o(1)}$ work and sublinear time. These downsides carry over to any algorithm that requires max-flow as a black box, and hence to the cut-based problems mentioned above.

To address these issues, Li~et~al.~\cite{li2023near} introduced the concept of \emph{fair cuts}, a robust form of \emph{approximate} min-cut. They present an algorithm for $(1+\epsilon)$-fair cut in $\tilde{O}(m/\epsilon^3)$ time that can be parallelized,\footnote{The parallelization requires $m^{1+o(1)}$ work and $m^{o(1)}$ time, but the $m^{o(1)}$ factors can be improved to polylogarithmic by recent work~\cite{agarwal2024parallel}. For simplicity, we do not discuss parallelization in this paper.} and then show how to solve $(1+\epsilon)$-approximate Steiner min-cut and Gomory-Hu tree using fair cut as a black box, leading to $\tilde{O}(m/\epsilon^{O(1)})$ time algorithms for both problems that can be parallelized. They also establish the first $\tilde{O}(m)$ time algorithm for expander decomposition that can also be parallelized.

The fair cut problem should be viewed as a generalization of $(1+\epsilon)$-approximate min-cut, which can be solved in $\tilde O(m/\epsilon)$ time by a recent breakthrough of Sherman~\cite{Sherman2017area} but is not robust enough for the above applications.\footnote{In more technical terms, the concept of \emph{uncrossing two cuts} breaks down for arbitrary approximate min-cuts. Fair cuts are designed to satisfy an approximate version of uncrossing, which suffices for the applications.} Nevertheless, there was a gap between the $\tilde{O}(m/\epsilon)$ time algorithm for $(1+\epsilon)$-approximate min-cut (and max-flow) and the $\tilde{O}(m/\epsilon^3)$ time algorithm for $(1+\epsilon)$-fair cut.

In this paper, we close the gap between the two problems by solving $(1+\epsilon)$-fair cut in $\tilde{O}(m/\epsilon)$ time. Conceptually, we present evidence that fair cut is no harder than approximate min-cut despite being more robust and powerful.

\begin{theorem}\label{thm:main}
There is an $\tilde{O}(m/\epsilon)$ time randomized algorithm that, with high probability,\footnote{We adopt the convention that \emph{with high probability} means with probability $1-1/n^{O(1)}$ for arbitrarily large polynomial in $n$.} solves $(1+\epsilon)$-fair cut on an undirected graph with integral and polynomial capacities.
\end{theorem}

Our algorithm is iterative, sending flow on each iteration and updating the \emph{residual} graph, which is directed. Our main idea is observing that Sherman's approximate max-flow/min-cut algorithm (for undirected graphs) actually performs well on certain \emph{directed} graphs, such as residual graphs of originally undirected graphs. 

\subsection{Preliminaries}

We work with both undirected and directed graphs in this paper. For an undirected graph $G=(V,E)$, let $\overleftrightarrow G=(V,\overleftrightarrow E)$ be the directed version of $G$ with each edge replaced by bidirectional arcs of the same capacity. Given a vertex set $S\subseteq V$, let $\partial S$ in an undirected graph be the set of edges with exactly one endpoint in $S$, and let $\vec\partial S$ in a directed graph be the set of arcs whose tail is in $S$ and whose head is not in $S$. We may also use $\vec\partial S$ for an undirected graph $G$, in which case we are referring to the bidirected $\overleftrightarrow G$.

Throughout the paper, we use $c_G(\cdot)$ to denote edge and arc capacities. For an arc/edge set $F$, let $c_G(F)$ denote the total capacity of arcs/edges in $F$. We assume that all capacities are integers and polynomially bounded; in general, we would incur extra $\log W$ terms where $W$ is the maximum integral capacity, but we stick with polynomially bounded for simplicity.

We represent a \emph{flow} as a nonnegative vector $f\in\R^{\overleftrightarrow E}$ for an undirected graph and $f\in\R^E$ for a directed graph. The \emph{congestion} of a flow is $\max_{(u,v)\in\overleftrightarrow E}f(u,v)/c_G(u,v)$, where $\overleftrightarrow E$ is replaced by $E$ for a directed graph. Sometimes we abuse notation and say the flow \emph{has congestion} $\kappa$ if the congestion of the flow is \emph{at most} $\kappa$. A flow is \emph{feasible} if its congestion is at most $1$. A \emph{demand} is a vector $d\in\R^V$ with $\sum_vd_v=0$. A flow \emph{satisfies} or \emph{routes} demand $d\in\R^V$ if for each vertex $v\in V$, $\sum_{(v,w)}f(v,w)-\sum_{(u,v)}f(u,v)=d_v$, i.e., the net flow out of $v$ is exactly $d_v$. A flow is an $(s,t)$-flow of \emph{value} $\tau$ if it satisfies demand $\tau(\1_s-\1_t)$. Here, $\1_v$ is the vector with entry $1$ at $v$ and entry $0$ elsewhere. We also use $\1$ as the all-ones vector.

Given an undirected/directed graph $G=(V,E)$ and a flow $f$, the \emph{residual} graph $G'$ of $G$ for flow $f$ is the directed graph with arc capacities $c_{G'}(u,v)=c_G(u,v)-f(u,v)+f(v,u)$ for each $(u,v)$ where either $(u,v)\in E$ or $(v,u)\in E$. Here, $c_G(u,v), f(u,v)$ are zero if $(u,v)\notin E$ and likewise for $(v,u)$.

A \emph{cut} is a bipartition $(S,V\setminus S)$ of the vertex set where $S,V\setminus S\ne\emptyset$. It is an $(s,t)$-cut if $s\in S$ and $t\notin S$. For a directed graph $G$, the \emph{value} of the cut $(S,V\setminus S)$ is $c_G(\vec\partial S)$. We require the following fact about the \emph{submodularity} of the directed cut function: for any directed graph $G$ and two sets $A,B\subseteq V$, $c_G(\vec\partial A)+c_G(\vec\partial B)\ge c_G(\vec\partial(A\cup B))+c_G(\vec\partial(A\cap B))$.

We now define the object of study in this paper, a \emph{fair cut}.

\begin{definition}
Let $s,t$ be two vertices in $V$. For any parameter $\alpha\ge 1$,
we say that a cut $(S,V\setminus S)$ is an \emph{$\alpha$-fair $(s,t)$-cut}
if there exists a feasible $(s,t)$-flow $f$ such that $f(u,v)\ge\frac{1}{\alpha}\cdot c(u,v)$
for every arc $(u,v)\in \vec\partial S$.
\end{definition}

We defer the matrix notation from Sherman's approximate max-flow/min-cut algorithm to its relevant \Cref{sec:approximate-max-flow}.

\section{Fair Cut Algorithm}
In this section, we present our fair cut algorithm, establishing \Cref{thm:main}. It will be more convenient to prove the following version, where $\epsilon$ is replaced by $O(\epsilon\log n)$.
\begin{theorem}\label{thm:main-convenient}
There is an $\tilde{O}(m/\epsilon)$ time randomized algorithm that, with high probability, solves\linebreak $(1+O(\epsilon\log n))$-fair cut on an undirected graph with integral and polynomial capacities.
\end{theorem}

We will use the following approximate max-flow/min-cut primitive for \emph{residual} graphs, which we present in \Cref{sec:approximate-max-flow}.
\begin{restatable}{theorem}{STCutFlow}\label{thm:s-t-cut-flow}
Given an undirected graph $G$, a residual graph $G'$ of $G$, two vertices $s,t$, and a parameter $\tau>0$, there is a randomized algorithm that runs in time $\Tilde{O}(m/\epsilon)$ and computes, with high probability,
 \begin{enumerate}
 \item Either an $(s,t)$-cut of value less than $\tau$, or
 \item A feasible flow $f$ in $G'$ routing a demand $d$ such that the residual demand $\tau(\mathbbm 1_s-\mathbbm 1_t)-d$ can be routed in $G$ with congestion~$\epsilon$.
 \end{enumerate}
\end{restatable}

Equipped with this flow/cut primitive, the fair cut algorithm is quite simple. We iteratively maintain a cut $(S_i,V\setminus S_i)$ and a flow $f_i$ that both gradually improve over time. On each iteration, we temporarily remove the edges in $\partial S_i$ that are nearly saturated in the right direction, and then call the flow primitive on the residual graph (minus the removed edges) with a careful choice of $\tau$. Depending on whether the flow primitive returns a flow or cut, we either update the current flow or the current cut, leaving the other unchanged.

We present the formal algorithm below.

\begin{enumerate}
\item Let $f_1$ be the empty flow and $(S_1,V\setminus S_1)$ be an arbitrary $(s,t)$-cut.
\item For iteration $i=1,2,\ldots,L=\Theta(\log n)$:
 \begin{enumerate}
 \item Let $\vec U_i\subseteq\overleftrightarrow E$ be all arcs $(u,v)\in\vec\partial S_i$ satisfying $f_i(u,v)\le(1-4\epsilon)c_G(u,v)$, i.e., the ``unsaturated'' arcs in $\vec\partial S_i$.
 \item Let $U_i\subseteq E$ be $\vec U_i$ with all arcs undirected (and parallel edges removed).
 \item Let $G_i\subseteq G$ be the undirected graph $G\setminus(\partial S_i\setminus U_i)$, i.e., remove all edges in $\partial S_i$ that are ``saturated'' in the right direction.
 \item Let $G'_i$ be the residual graph of $G_i$ for the restricted flow $f_i|_{G_i}$, defined as the flow $f_i$ with flow on arcs outside $\overleftrightarrow G_i$ removed.
 \item Call \Cref{thm:s-t-cut-flow} on graph $G_i$, its residual graph $G'_i$, vertices $s,t$, and parameter $\tau=0.5c_{G'_i}(\vec\partial S_i)$.
 \item If \Cref{thm:s-t-cut-flow} returns a flow $h$:
  \begin{itemize}
  \item Set $f_{i+1}=f_i+h$ and $S_{i+1}=S_i$, i.e., add the new flow but keep the current cut.
  \end{itemize}
 \item If \Cref{thm:s-t-cut-flow} returns a cut $(X_i,V\setminus X_i)$:
  \begin{itemize}
  \item Set $f_{i+1}=f_i$ and $S_{i+1}$ as either $S_i\cup X_i$ or $S_i\cap X_i$, i.e., update the cut but keep the current flow. Of the two choices, pick the $S_{i+1}$ minimizing $c_{G'_i}(\vec\partial S_{i+1})$.
  \end{itemize}
 \end{enumerate}
 \item Output the $(s,t)$-cut $(S_{L+1},V\setminus S_{L+1})$.
\end{enumerate}
It is clear that the algorithm makes $O(\log n)$ calls to \Cref{thm:s-t-cut-flow} and runs in $\tilde{O}(m)$ time outside these calls, for an overall running time of $\tilde{O}(m/\epsilon)$. For the rest of this section, we prove its correctness.

Our measure of progress is the quantity $c_{G'_i}(\vec\partial S_i)$, i.e., the total residual capacity of all ``unsaturated'' arcs in $\vec\partial S_i$, which we show drops by a constant factor on each iteration.

\begin{lemma}\label{lem:progress-flow}
If \Cref{thm:s-t-cut-flow} returns a flow $h$, then $c_{G'_{i+1}}(\vec\partial S_{i+1})\le0.75c_{G'_i}(\vec\partial S_i)$.
\end{lemma}
\begin{proof}
Let $d$ be the demand routed by flow $h$. By \Cref{thm:s-t-cut-flow}, there is a feasible flow $r$ in $G_i$ routing the residual demand $\tau(\mathbbm 1_s-\mathbbm 1_t)-d$ with congestion~$\epsilon$. Then, the flow $h+r$ in $G'_i$ routes demand $\tau(\mathbbm 1_s-\mathbbm 1_t)$ (with arbitrary congestion). This flow must pass through the cut $\vec\partial S_i$ in $G'_i$, so $h+r$ sends a net flow of $\tau$ across $\vec\partial S_i$. Since the flow $r$ has congestion $\epsilon$ in $G_i$, removing it from $h+r$ affects the net flow across $\vec\partial S_i$ by at most $\epsilon c_{G_i}(\partial S_i)$, so the flow $h$ sends at least $\tau-\epsilon c_{G_i}(\partial S_i)$ across $\vec\partial S_i$. Each arc $(u,v)\in\vec\partial S_i$ in $G'_i$ satisfies $(u,v)\in U_i$, so $f_i(u,v)\le(1-4\epsilon)c_G(u,v)$ and $c_{G'_i}(u,v)\ge 4\epsilon c_G(u,v)=4\epsilon c_{G_i}(u,v)$. Summing over all arcs $(u,v)\in\vec\partial S_i$ gives $c_{G'_i}(\vec\partial S_i)\ge4\epsilon c_{G_i}(\partial S)$, so the flow $h$ sends at least $\tau-0.25c_{G'_i}(\vec\partial S_i)=0.25c_{G'_i}(\vec\partial S_i)$ flow across $\vec\partial S_i$.

Let $H$ be the residual graph of $G_i$ for the restricted flow $f_{i+1}|_{G_i}$. By definition of residual graph, the quantity $c_{G_i}(\partial S_i)-c_H(\vec\partial S_i)$ is exactly the net flow that $f_{i+1}|_{G_i}$ sends across $\vec\partial S_i$. Since $G'_i$ is the residual graph of $G_i$ for the restricted flow $f_i|_{G_i}$, the quantity $c_{G_i}(\partial S_i)-c_{G'_i}(\vec\partial S_i)$ is exactly the net flow that $f_i|_{G_i}$ sends across $\vec\partial S_i$.   Since $f_{i+1}|_{G_i}=f_i|_{G_i}+h$, the difference of quantities ${(c_{G_i}(\partial S_i)-c_H(\vec\partial S_i))-(c_{G_i}(\partial S_i)-c_{G'_i}(\vec\partial S_i))}$ is exactly the net flow that $h$ sends across $\vec\partial S_i$, which is at least $0.25c_{G'_i}(\vec\partial S_i)$. In other words, $c_H(\vec\partial S_i)\le c_{G'_i}(\vec\partial S_i)-0.25c_{G'_i}(\vec\partial S_i)=0.75c_{G'_i}(\vec\partial S_i)$. Any previously ``saturated'' arc $(u,v)\in\vec\partial S_i\setminus\vec U_i$ is still ``saturated'' in flow $f_{i+1}$ (i.e., $f_{i+1}(u,v)>(1-4\epsilon)c_G(u,v)$) since the arc is absent from $G'_i$ and hence carries no flow in $h$. Since $S_{i+1}=S_i$, we have $\partial S_i\setminus U_i\subseteq\partial S_{i+1}\setminus U_{i+1}$, which means that $G_i\supseteq G_{i+1}$. In particular, the arcs in $\vec\partial S_i$ present in $G'_{i+1}$ are also present in $H$, and they have the same capacity since both $G'_{i+1}$ and $H$ are residual graphs for a restriction of $f_{i+1}$. We conclude that $c_{G'_{i+1}}(\vec\partial S_i)\le c_H(\vec\partial S_i)$, and
\[ c_{G'_{i+1}}(\vec\partial S_{i+1})=c_{G'_{i+1}}(\vec\partial S_i)\le c_H(\vec\partial S_i)\le0.75c_{G'_i}(\vec\partial S_i) ,\]
as promised.
\end{proof}

\begin{lemma}\label{lem:progress-cut}
If \Cref{thm:s-t-cut-flow} returns a cut $(X_i,V\setminus X_i)$, then $c_{G'_{i+1}}(\vec\partial S_{i+1})\le0.75c_{G'_i}(\vec\partial S_i)$.
\end{lemma}
\begin{proof}
By \Cref{thm:s-t-cut-flow}, the cut $(X_i,V\setminus X_i)$ satisfies $c_{G'_i}(\vec\partial X_i)\le\tau=0.5c_{G'_i}(\vec\partial S_i)$. By submodularity of the cut function $c_{G'_i}(\vec\partial S)$,
\[ c_{G'_i}(\vec\partial S_i)+c_{G'_i}(\vec\partial X_i)\ge c_{G'_i}(\vec\partial(S_i\cup X_i))+c_{G'_i}(\vec\partial(S_i\cap X_i)) ,\]
and by the choice of $S_{i+1}$,
\[ c_{G'_i}(\vec\partial S_{i+1})\le\frac12\big( c_{G'_i}(\vec\partial(S_i\cup X_i))+c_{G'_i}(\vec\partial(S_i\cap X_i))\big)\le\frac12\big(c_{G'_i}(\vec\partial S_i)+c_{G'_i}(\vec\partial X_i)\big)\le0.75c_{G'_i}(\vec\partial S_i).\]

We now claim that $c_{G'_{i+1}}(\vec\partial S_{i+1})\le c_{G'_i}(\vec\partial S_{i+1})$. Note that arcs present in both $G'_i$ and $G'_{i+1}$ must have the same capacity since $f_{i+1}=f_i$, so it suffices to show that the arcs in $\vec\partial S_{i+1}$ present in $G'_{i+1}$ are a subset of those present in $G'_i$. Any arc $(u,v)\in\vec\partial S_{i+1}$ present in $G'_{i+1}$ satisfies $(u,v)\in\vec U_{i+1}$, so $f_i(u,v)=f_{i+1}(u,v)\le(1-4\epsilon)c_G(u,v)$. If $(u,v)\in\vec\partial S_i$ as well, then $(u,v)\in\vec U_i$ and the arc belongs to $G'_i$. Otherwise, if $(u,v)\notin\vec\partial S_i$, then there are two cases. If $S_{i+1}=S_i\cap X_i$, then since $u\in S_{i+1}\subseteq S_i$, we must have $v\in S_i$ as well. So the edge $(u,v)$ is not in $\partial S_i$, which means the arc $(u,v)$ belongs to $G'_i$, establishing the claim. If $S_{i+1}=S_i\cup X_i$, then since $v\in V\setminus S_{i+1}\subseteq V\setminus S_i$, we must have $u\in V\setminus S_i$ as well. So the edge $(u,v)$ is not in $\partial S_i$, and the same argument follows.

Putting everything together, we conclude that $c_{G'_{i+1}}(\vec\partial S_{i+1})\le c_{G'_i}(\vec\partial S_{i+1})\le0.75c_{G'_i}(\vec\partial S_i)$.
\end{proof}

Finally, we prove the correctness of the algorithm, establishing \Cref{thm:main-convenient}.
\begin{lemma}
The output $(S_{L+1},V\setminus S_{L+1})$ is a $(1+O(\epsilon\log n))$-fair cut with high probability.
\end{lemma}
\begin{proof}
By \Cref{lem:progress-flow,lem:progress-cut}, we have $c_{G'_{i+1}}(\vec\partial S_{i+1})\le0.75c_{G'_i}(\vec\partial S_i)$ for each iteration $i$. Since capacities are polynomially bounded, we start with $c_{G'_1}(\vec\partial S_1)\le n^{O(1)}$, so for large enough $L=\Theta(\log n)$ we have $c_{G'_{L+1}}(\vec\partial S_{L+1})<4\epsilon$ with high probability. Any arc $(u,v)\in\vec U_{L+1}$ belongs to $\vec\partial S_{L+1}$ and satisfies $c_{G'_{L+1}}(u,v)\ge 4\epsilon c_G(u,v)\ge4\epsilon$, so no such arcs exist. In other words, $\vec U_{L+1}=\emptyset$, and it follows that $f_{L+1}(u,v)\ge(1-4\epsilon)c_G(u,v)$ for all arcs $(u,v)\in\vec\partial S_{L+1}$. To establish fairness, it remains to augment $f_{L+1}$ to an $(s,t)$-flow.

By construction, $f_{L+1}=h_1+h_2+\cdots+h_L$, and there exist flows $r_1,\ldots,r_L$ in $G$ of congestion $\epsilon$ such that each $h_i+r_i$ is an $(s,t)$-flow. In particular, the flow $f'=f_{L+1}+r_1+\cdots+r_L$ is an $(s,t)$-flow. Since $r_1+\cdots+r_L$ has congestion $O(\epsilon\log n)$, we have $|f'(u,v)-f_{L+1}(u,v)|\le O(\epsilon\log n)\cdot c_G(u,v)$ for all arcs $(u,v)$. In particular, for each arc $(u,v)\in\vec\partial S_{L+1}$,
\begin{align*}
f'(u,v)&\ge f_{L+1}(u,v)-O(\epsilon\log n)\cdot c_G(u,v)\\&\ge(1-4\epsilon)c_G(u,v)-O(\epsilon\log n)\cdot c_G(u,v)\\&\ge\frac1{1+O(\epsilon\log n)}c_G(u,v) ,
\end{align*}
so the $(s,t)$-flow $f'$ certifies that $(S_{L+1},V\setminus S_{L+1})$ is a $(1+O(\epsilon\log n))$-fair cut.
\end{proof}

\section{Approximate Max-Flow on Residual Graphs}\label{sec:approximate-max-flow}
In this section, we show how to apply Sherman's approximate max-flow/min-cut algorithm on directed \emph{residual graphs} of an underlying undirected graph. The flow may not satisfy the input demand, but the leftover demand will be routable with low congestion on the \emph{undirected} graph. Our main goal is to prove \Cref{thm:s-t-cut-flow}, restated below.

\STCutFlow*


We first introduce some preliminaries from Sherman~\cite{Sherman2017area}. For a matrix $A\in\R^{n\times m}$, define the matrix norm $\|A\|_{\infty \to \infty}=\displaystyle\max_{\|v\|_\infty=1}\|Av\|_{\infty}$, and define $\nnz(A)$ as the number of nonzero entries in $A$. Define $\Delta_k^m\subseteq\R^{m\times k}$ as the set of matrices $X\in\R^{m\times k}$ with $X\ge\mathbf0$ and $\sum_jX_{ij}=1$ for all $i\in[m]$. We now present a key subroutine from~\cite{Sherman2017area}:
\begin{theorem}[Corollary~1.8 of~\cite{Sherman2017area}]\label{thm:Sherman-1.8}
There is an algorithm that, given $B \in \R^{n\times k}$, and $A \in \R^{n\times m}$ with $\|A\|_{\infty \to \infty} \le 1$, takes $\Tilde{O}(k \nnz(A) \epsilon^{-1})$ time and outputs either, 
\begin{enumerate}[(1)]
    \item $X \in \Delta_k^m$ such that $AX \le B + \epsilon R$ where $R \in \Delta_k^n$.
    \item $Y \in \R^{k\times n}, Y \ge \mathbf0$ such that $\tr(Y(AX-B)) > 0$ for all $X \in \Delta_k^m$.
\end{enumerate}
\end{theorem}

This result can be used to solve approximate multi-commodity flow, as indicated by Lemma~4.2 of~\cite{Sherman2017area}. Sherman only provided a sketch proof in the original paper. For specificity and completeness, we state and prove a one-commodity flow version of Sherman's Lemma~4.2 for some residual network $G'$ of $G$, where the matrix $R$ is a congestion approximator for $G$. 

For a given directed graph $G=(V,\vec E)$, we represent a flow $f$ by its vector of congestions on edges, which is a nonnegative vector in $\R^{\vec E}$ where for each arc $(u,v)\in\vec E$, the flow $f$ sends $f_{(u,v)}c_G(u,v)$ flow. Note that $\|f\|_\infty$ is exactly the congestion of the flow. We define $C_G\in\R^{\vec E\times\vec E}$ as the diagonal matrix whose entry $(u,v)$ is the capacity of arc $(u,v)\in\vec E$. We define $D_G\in\R^{\vec E\times V}$ as the matrix whose row $(u,v)\in\vec E$ has vector $(\mathbbm 1_u-\mathbbm 1_v)^\top$, also called the discrete divergence operator which maps any vector $C_Gf$ to the demand satisfied by the flow $f$. If $G=(V,E)$ is an undirected graph, we treat it as a directed graph with bidirected edge set $\overleftrightarrow E$. 

For any \emph{undirected} graph $G$, an $\alpha$-congestion approximator is a matrix $R\in\R^{[r]\times V}$ (where dimension $r$ is unspecified) such that for any demand $d$ whose optimal flow has congestion $OPT(d)$, it holds that $\|Rd\|_\infty\le OPT(d)\le\alpha\|Rd\|_\infty$.



    Given a directed graph $G = (V,E)$, we say $G' = (V, E')$ is a subgraph of $G$ (denoted as $G' \le G$) iff for all $(u,v) \in V\times V$, it holds that $c_{G'}(u,v) \le c_G(u,v)$. 
In particular, if $G'$ is some residual network for undirected $G$, then it holds that $G' \le 2G$.

\begin{lemma}\label{lem:operator-norm}
    Given an undirected graph $G$, a subgraph $G' \le\overleftrightarrow G$, and an $\alpha$-congestion approximator $R$ of $G$, it holds that $\|R D_{G'} C_{G'}\|_{\infty \to \infty} \le 2$.
\end{lemma}

\begin{proof}
Fix any vector $f \in \R^{\overleftrightarrow E}$ with $\|f\|_\infty = 1$, and let $d := D_{\overleftrightarrow G} C_{\overleftrightarrow G} f$. Notice that each pair of anti-parallel arcs $e^+, e^-$ of $\overleftrightarrow G$ has the same capacity, and adding a constant to $f_{e^+}, f_{e^-}$ will not change the demand routed by the flow; we perform such operations on each pair of anti-parallel arcs to obtain some $f'$ such that $\min(f_{e^+}', f_{e^-}') = 0$ for all pairs of anti-parallel arcs. Now it is evident that $\|f'\|_\infty \le 2$, and that $f'$ corresponds to a flow $f_G$ of $G$ with demand $d$ and congestion $\le 2$, so 
    $$\|R D_{G} C_{G} f\|_\infty = \|Rd\|_\infty \le OPT(d) \le \|f_G\|_\infty = 2$$
    and it follows that 
    $$\|RD_G C_G f\|_{\infty} = \|Rd\|_\infty \le 2$$
    so
    $$\|R D_G C_G\|_{\infty \to \infty} \le 2.$$
    Let $F := \{f \in \R^{\overleftrightarrow E}\ :\ \|f\|_\infty \le 1\}$ be the set of all directed flows of congestion $1$. Since $c_{G'}(u,v) \le c_G(u,v)$ for all arcs $(u,v)\in\overleftrightarrow E$, we have the set inclusion
    $$\{ D_{G'} C_{G'}f : f\in F\} \subseteq \{ D_G C_Gf : f \in F\}$$
    so
    $$\max_{f\in F}\|R D_{G'} C_{G'}f\|_\infty \le \max_{f\in F}\|R D_G C_G f\|_\infty \le 2$$
    and it follows that $\|R D_{G'} C_{G'}\|_{\infty \to \infty} \le 2$.
    \end{proof}

Note that Sherman's Stochastic Matrix Algorithm, as per Corollary 1.8 of \cite{Sherman2017area}, may output a dual, and when leveraging Sherman's algorithm to solve approximate $(s,t)$-max-flow, we may need to explicitly compute a corresponding $(s,t)$-cut that is integral. The following lemma constructs such a cut in a directed graph, given an infeasibility criterion from Sherman's algorithm: 

\begin{lemma}\label{lem:threshold-cut}
    Given a directed graph $G$ and a ``potential'' vector $\phi \in \R^n$ on vertices, we define a corresponding flow $f_\phi$ via the following:
    $$(f_\phi)_{uv} = \begin{cases}
        1, \text{ if } \phi_u > \phi_v\\
        0, \text{ otherwise }
    \end{cases}$$
    We also suppose that 
    $$\phi^\top(d-B_GC_Gf_\phi) > 0$$ 
    for some demand $d$. Then if we sort the vertices by decreasing potential $\phi_v$, there must be some prefix $S \subset V$, such that $\sum_{v \in S} d_v > c_G(\vec\partial S)$, certifying the infeasibility of such a demand. Furthermore, the cut $(S, V\setminus S)$ can be computed in $O(m+n\log n)$ time. 

Moreover, if $d=\tau(\mathbbm 1_s-\mathbbm 1_v)$ for two vertices $s,t\in V$ and parameter $\tau>0$, then the cut $(S,V\setminus S)$ is an $(s,t)$-cut with $c_G(\vec\partial S)<\tau$.
\end{lemma}

\begin{proof}
We begin with some notation. Let
$V_{>x} := \{v \in V(G)\ :\ \phi_v > x\}$
denote the set of vertices of $G$ whose potential is strict greater than $x$,
and let $\Delta(S) := \sum_{v \in S} d_v$
denote the sum of demands in the set $S$ of vertices.

    Let $M$ be some positive real number such that $|\phi_v| < M$ for all $v \in V(G)$. We seek to prove
    $$\int_{-M}^M \Delta(V_{>x}) dx = \phi^\top d > \phi^\top(B_GC_Gf_\phi) = \int_{-M}^M c_G(\vec\partial V_{>x}) dx$$
    because then there must be some $-M \le x \le M$ s.t. $\Delta(V_{>x}) > c_G(\vec\partial V_{>x})$, and setting $S=V_{>x}$ achieves the desired $\Delta(S)> c_G(\vec\partial S)$.

\begin{enumerate}
 \item
    For the first equality, it holds that 
    \begin{align*}
        \int_{-M}^M \Delta(V_{>x}) dx = \int_{-M}^M \sum_{v\in V} d_v\cdot \mathbbm 1[\phi_v > x] dx =  \sum_{v\in V} \int_{-M}^M d_v\cdot \mathbbm 1[\phi_v > x] dx = \sum_{v\in V} d_v (\phi_v + M)
    \end{align*}
    and since $d$ is a demand, it holds that $\sum_v d_v = 0$, so the above equals
    $$ \sum_{v\in V} d_v \phi_v = \phi^\top d$$
 \item
    For the second equality, we start from the definition of $f_\phi$ and have 
    \begin{align*}
        \phi^\top(B_GC_Gf_\phi) &= \sum_v \phi_v \cdot (-\sum_{(u,v)\in E} c_G(u,v) \mathbbm 1[\phi_u > \phi_v] + \sum_{(v,w) \in E} c_G(v,w) \mathbbm 1[\phi_v > \phi_w] ) \\
        &= \sum_{(u,v) \in E} c_G(u,v) \mathbbm 1(\phi_u > \phi_v) (\phi_u - \phi_v)\\
        &= \sum_{(u,v) \in E} c_G(u,v) \max(0, \phi_u - \phi_v)\\
        &= \int_{-M}^M \sum_{(u,v) \in E} c_G(u,v) \mathbbm 1((u,v) \in \vec\partial (V_{>x}))dx\\
        &= \int_{-M}^M c_G(\vec\partial V_{>x}) dx
    \end{align*}
 \item 
    The inequality follows from the assumption $\phi^\top(d-B_GC_Gf_\phi) > 0$.
\end{enumerate}
    Furthermore, if $d=\tau(\mathbbm 1_s-\mathbbm 1_t)$, then notice that
$$\Delta(S)> c_G(\vec\partial S)\ge0$$
    and the only way for $\Delta(S)>0$ is if $s\in S$ and $t\notin S$, in which case $\Delta(S)=\tau$, so $(S,V\setminus S)$ is an $(s,t)$-cut with $c_G(\vec\partial S)<\tau$.

    To identify the set $S=V_{>x}$, it suffices to iterate through the cuts corresponding to all $O(n)$ prefixes $S=V_{>x}$, keeping track of $\Delta(S)$ and $c_G(\vec\partial S)$. We may start with $|S|=1$, which includes the vertex $v$ with the highest potential $\phi_v$, identify edges in the cut, and compute the value of $c_G(\vec\partial S)$; then we keep adding vertices to $S$, one at a time, in the order of decreasing potentials, and for each added vertex, we re-compute the set of edges in the new cut, then compute the associated capacity. Each update after adding some vertex $v$ requires iterating through $\deg(v)$ many edges, and since $\sum_v \deg(v) = 2m$, the overall time complexity to find the minimum threshold cut is $O(m)$ (after the initial sorting by $\phi_v$ in $O(n\log n)$).
\end{proof}

Using \Cref{lem:operator-norm,lem:threshold-cut}, we finally prove \Cref{thm:s-t-cut-flow}:

\begin{proof}
    To convert a single commodity flow into a right stochastic matrix problem, we consider the ``empty demand'' as another type of commodity. Specifically, we fix $k=2$ and $X = [f, f_\emptyset]$, where $f$ is the vector of a flow of congestion $1$, with its entries indicating congestion on edges, and $f_\emptyset$ is the vector indicating remaining congestion, defined as $f_\emptyset := \mathbbm 1 - f$; we also define $A := D_{G'}C_{G'}$, and $B := [d, d_\emptyset]$, where $d_\emptyset := D_G C_G \mathbbm{1} - d$ is the empty demand vector. Then $AX = B$ encodes a solution for the single commodity problem. 

Next, we cite Theorem 4.4 from Sherman \cite{Sherman2017area} to compute an $\alpha$-congestion approximator $R$ of $G$ with $\alpha=(\log n)^{O(1)}$ in $\tilde O(m)$ time with high probability, where $R$ is specifically column sparse, i.e.\ each column contains $O(\log n)$ many nonzero entries. Observe that $(RD_G' C_G')_{r,e}\ne0$ iff the directed edge $e$ is in the cut represented by row $r$ of $R$, and since $R$ is column sparse, there can only be $O(\log n)$ many cuts in $R$ that contain $e$; overall we obtain 
    $$nnz(RD_G' C_G') \le O(m \log n)$$

    Since $G'$ is a residual graph of $G$, we have $G'\le2G$, and applying \Cref{lem:operator-norm} on subgraph $G'/2$ (and then scaling up by factor $2$) gives $\|R D_{G'} C_{G'}\|_{\infty \to \infty} \le 4$. Now define $R' := R/4$ and
$A_2 := \begin{bmatrix}R'A \\ -R'A\end{bmatrix}$ and $B_2:= \begin{bmatrix}B \\ -B\end{bmatrix}$, so that 
$$ \|A_2\|_{\infty \to \infty}=\|R'A\|_{\infty\to\infty}=\|RA/4\|_{\infty\to\infty}=\|RD_{G'}C_{G'}/4\|_{\infty \to \infty} \le 1 $$
 and we invoke \Cref{thm:Sherman-1.8} to obtain, in time $\tilde{O}(k\,\nnz(A_2)\epsilon^{-1})\le\tilde{O}(k\,\nnz(RD_{G'}C_{G'})\epsilon^{-1})\le\tilde{O}(m\epsilon^{-1})$, either
\begin{enumerate}
\item some $X$ and some $S_1, S_2 \in \Delta_2^n$ such that
    $$\begin{cases}
    R' A X - R' B \le \frac{\epsilon}{\alpha} S_1 \\
    R' (-A) X - R' (-B) \le \frac{\epsilon}{\alpha} S_2
    \end{cases}$$ 
\item some $Y = [y_1, y_2] > 0$ such that
$$\text{tr}(Y(A_2X-B_2)) \ge 0$$
for all feasible flows $X = [f, f_\emptyset]$
\end{enumerate}
    In the first case, combining the two inequalities gives
    $$-\frac{\epsilon}{\alpha} S_2 \le R'(AX-B) \le \frac{\epsilon}{\alpha} S_1$$
    thus each row $r_i\in\R^2$ of $R'(AX-B)$ satisfies $\|r_i\|_{\infty} \le \frac{\epsilon}{\alpha}$, so that $\|r_i\|_1 \le 2\frac{\epsilon}{\alpha}$, and 
    \begin{align*}
        \|R'(Af-d)\|_{\infty\to\infty} &= \max_{\|v\|_\infty = 1} \|R'(Af-d) v\|_\infty\\
        &= \max_{\|v\|_\infty = 1} \max_i |(R'(Af-d) v)_i|\\
        &= \max_i \max_{\|v\|_\infty = 1} |(R'(Af-d) v)_i|\\
        &= \max_i \|(R'(Af-d))_i\|_1\\
        &\le \max_i \|(R'(AX-B))_i\|_1\\
        &\le \frac{2\epsilon}{\alpha}
    \end{align*}
    and it follows from definition of $R$ that if the algorithm returns some flow $X = [f, f_\emptyset]$, then the residual demand $Af-d$ can be routed with congestion $\le O(\epsilon)$ with respect to the undirected graph $G$. We may retroactively set $\epsilon$ a constant factor smaller so that the congestion is $\le\epsilon$.

    In the second case, it equivalently holds for all feasible flows $X = [f, f_\emptyset]$ that     
    $$y_1^\top(\begin{bmatrix}A \\ -A\end{bmatrix} f- \begin{bmatrix}d \\ -d\end{bmatrix}) + y_2^\top(\begin{bmatrix}A \\ -A\end{bmatrix}f_\emptyset - \begin{bmatrix}d_\emptyset \\ -d_\emptyset\end{bmatrix}) > 0$$ 
    By construction of duplicated rows, we have $y_i = \begin{bmatrix}w_i \\ z_i\end{bmatrix}$, and either one of 
    \begin{enumerate}[(i)]
        \item $w_1^\top (Af-d) + w_2^\top (Af_\emptyset-d_\emptyset) > 0$
        \item $z_1^\top (Af-d) + z_2^\top (Af_\emptyset-d_\emptyset) < 0$
    \end{enumerate}
    is true. We further notice that 
    \begin{align*}
        A(f+f_\emptyset) = d+d_\emptyset \iff
        Af-d = -(Af_\emptyset - d_\emptyset)
    \end{align*}
    so we substitute and either one of 
    \begin{enumerate}[(i)]
        \item $(w_2-w_1)^\top(d-Af) > 0$
        \item $(z_1-z_2)^\top(d-Af) > 0$
    \end{enumerate}
    is true; either of which, according to \Cref{lem:threshold-cut}, implies the existence of some $(s,t)$-cut $(S,V\setminus S)$ with $c_G(\vec\partial S)<\tau$ that can be computed in $O(m+n\log n)$ time.
\end{proof}

\bibliographystyle{alpha}
\bibliography{main}

\end{document}